\documentclass[a4paper,UKenglish,cleveref, autoref, thm-restate]{lipics-v2021}

\usepackage[ruled,vlined,linesnumbered,commentsnumbered]{algorithm2e}
\usepackage{xcolor}
\usepackage{hyperref}
\usepackage{algcompatible}
\usepackage{amsmath,amssymb,amsfonts}
\usepackage{graphicx}

\nolinenumbers

\title             {Improved bicriteria approximation for {\em k}-edge-connectivity}
\titlerunning{Improved bicriteria approximation for {\em k}-edge-connectivity}

\author{Zeev Nutov}{The Open University of Israel}{nutov@openu.ac.il}{https://orcid.org/0000-0002-6629-3243}{}

\authorrunning{Zeev Nutov}
\Copyright{Zeev Nutov}


\begin{document}

\maketitle

\newcommand {\ignore} [1] {}


\def\CC  {{\cal C}}
\def\LL  {{\cal L}}
\def\RR  {{\cal R}}
\def\TT  {{\cal T}}

\def\de  {\delta}
\def\th   {\theta}
\def\be  {\beta}
\def\th   {\theta}
\def\De  {\Delta}

\def\empt {\emptyset}
\def\sem  {\setminus}
\def\subs {\subseteq}

\def\f  {\frac}
\def\opt {\sf opt}
\def\span {\sf span}

\def\kECSS  {{\sc $k$-ECSS}}  
\def\kECSM {{\sc $k$-ECSM}}  
\def\kCSS {{\sc $k$-VCSS}}        

\keywords{$k$-edge-connected subgraph, bicriteria approximation, iterative LP-rounding}

\begin{abstract}
In the {\sc  $k$-Edge Connected Spanning Subgraph} ({\kECSS}) problem 
we are given a (multi-)graph $G=(V,E)$ with edge costs and an integer $k$, 
and seek a min-cost $k$-edge-connected spanning subgraph of $G$. 
The problem admits a $2$-approximation algorithm and no better approximation ratio is known.
Hershkowitz, Klein, and Zenklusen [STOC 24] gave a bicriteria $(1,k-10)$-approximation algorithm 
that computes a $(k-10)$-edge-connected spanning subgraph 
of cost at most the optimal value of a standard Cut-LP for {\kECSS}. 
This LP bicriteria approximation was recently improved by Nutov and Cohen \cite{CN} to $(1,k-4)$,
where also was given a bicriteria approximation $(3/2,k-2)$.
In this paper we improve the bicriteria approximation to 
$(1,k-2)$ for $k$ even and to
$\left(1-\f{1}{k},k-3\right)$ for $k$ is odd,
and also give another bicriteria approximation $(3/2,k-1)$.
After this paper was written, we became aware that the same result was achieved earlier by Kumar and Swamy \cite{KS}.

The {\sc $k$-Edge-Connected Spanning Multi-subgraph} ({\kECSM}) problem 
is almost the same as {\kECSS},
except that any edge can be selected multiple times at the same cost.
The previous best approximation ratio for {\kECSM} was $1+4/k$ \cite{CN}.
Our result improves this to  
$1+\f{2}{k}$ for $k$ even and to
$1+\f{3}{k}$ for $k$ odd,
where for $k$ odd the computed subgraph is in fact $(k+1)$-edge-connected.
\end{abstract}

\section{Introduction} \label{s:intro}

A graph is {\bf $k$-edge-connected} if it contains $k$ edge-disjoint paths between any two nodes.
We consider the following problem.

\begin{center} \fbox{\begin{minipage}{0.98\textwidth}
\underline{\sc  $k$-Edge-Connected Spanning Subgraph} ({\kECSS}) \\
{\em Input:} \ \ A (multi-)graph $G=(V,E)$ with edge costs and an integer $k$. \\
{\em Output:} A  min-cost $k$-edge-connected spanning subgraph of $G$. 
\end{minipage}} \end{center}

The problems admits approximation ratio $2$ by a reduction to a bidirected graph (c.f. \cite{K}), 
and no better approximation ratio is known.
Hershkowitz, Klein, and Zenklusen \cite{HKZ} 
gave a bicriteria $(1,k-10)$-approximation algorithm 
that computes a $(k-10)$-edge-connected spanning subgraph of cost
at most the optimal value of a standard Cut-LP for {\kECSS}:
\begin{equation} \label{LPk}
\displaystyle
\begin{array} {lllllll} 
& \min              & \displaystyle c^T \cdot x & \\
& \ \mbox{s.t.} & x(\de_E(S)) \ge k                & \forall \empt \ne S \subset V \\
&                      & 0 \le x_e \le 1                  & \forall e \in E                                                                 
\end{array}
\end{equation}
where $\de_E(S)$ is the set of edges in $E$ with exactly one end in $S$ 
and $x(\de_E(S))=\sum_{e \in \de_E(S)}x_e$ is the sum of their $x$-values.
This was improved to a $(1,k-4)$ by Nutov and Cohen \cite{CN},
where also was given a $(3/2,k-2)$ bicriteria approximation.  
We improve this as follows. 

\begin{theorem} \label{t:main}
{\kECSS} admits the following bicriteria approximation ratios w.r.t. the Cut-LP: 
$(1,k-2)$ for $k$ is even and $\left(1-\f{1}{k},k-3\right)$ for $k$ odd.
The problem also admits a bicriteria approximation ratio $(3/2,k-1)$.
\end{theorem}

Note that for $k$ odd the cost of the produced solution is $\left(1-\f{1}{k}\right)\opt$,
which is strictly less than $\opt$.
For example, for $k=5$ the algorithm computes a $2$-edge-connected subgraph of cost 
at most $4/5$ times the optimal cost of a $5$-edge-connected subgraph. 

After this paper was written we became aware that 
the same result was achieved earlier by Kumar and Swamy \cite{KS},
that also obtained additional results for degree-bounded versions of {\kECSS} and {\kECSM}.
But our work was done independently, and our proof is simpler than that of \cite{KS}. 

The {\sc $k$-Edge-Connected Spanning Multi-subgraph} ({\kECSM}) problem 
is almost the same as {\sc $k$-ECSS}, 
except that any edge can be selected multiple times at the same cost.
For {\kECSM}, approximation ratios $3/2$ for $k$ even and $3/2+O(1/k)$ for $k$ odd 
were known long time ago \cite{FJ1,FJ2}, and it was conjectured by Pritchard \cite{P} that 
{\kECSM} admits approximation ratio $1+O(1/k)$. 
Hershkowitz, Klein, and Zenklusen \cite{HKZ} showed that {\kECSM} 
has an $\left(1+\Omega(1/k)\right)$-approximation threshold 
and that a $(1,k-p)$ bicriteria approximation for {\kECSS} 
w.r.t. the standard Cut-LP implies approximation ratio  $1+p/k$ for {\kECSM}. 
Thus the novel result of \cite{HKZ} --  a $(1,k-10)$ bicriteria approximation for {\kECSS}, 
implies an approximation ratio $1+10/k$ for {\kECSM}
(improving the previous ratio $1+O(1/\sqrt{k})$ of  \cite{KKGZ}).
This was reduced to $1+4/k$ in \cite{CN}.
Using Theorem~\ref{t:main} we improve this as follows. 

\begin{corollary}
{\kECSM} admits approximation ratio $1+\f{2}{k}$ for $k$ even and $1+\f{3}{k}$ for $k$ odd,
where for $k$ odd the computed subgraph is $(k+1)$-edge-connected 
and has cost at most $1+\f{3}{k}$ times the cost of a $k$-connected subgraph.
\end{corollary}
\begin{proof}
Given an instance of $k$-ECSM, we multiply every edge $k+p$ times 
and compute a bicriteria approximate solution as in Theorem~\ref{t:main} to {\sc $(k+p)$-ECSS}, where 
$p=2$ if $k$ is even and $p=3$ if $k$ is odd; in both cases, $k+p$ is even, 
so the connectivity guarantee is $k+p-2$.
One can see that if $x$ is a feasible {\kECSM} LP solution
then $x \cdot \f{k + p}{k}$ is a feasible {\sc $(k+p)$-ECSM} LP solution.
Denoting by $\tau[\mbox{\sc $\ell$-ECSS}]$ and $\tau[\mbox{\sc $\ell$-ECSM}]$ 
the optimal value of an {\sc $\ell$-ECSS} and {\sc $\ell$-ECSM} LP solutions, respectively, we get:
\[
\tau[ \mbox{\sc $(k+p)$-ECSS}] = \tau[\mbox{\sc $(k+p)$-ECSM}] \le (1+p/k) \tau[\mbox{\sc $k$-ECSM}] \ .
\]
If $k$ odd then $k+p-2=k+1$, hence the returned subgraph is $(k+1)$-edge-connected.
\end{proof}

We highlight the main difference between 
our algorithm and those in previous works \cite{HKZ,CN}.
The algorithm of \cite{HKZ} maintains a set $I$ of integral edges and a set of $F$ of fractional edges.
It repeatedly finds an extreme point optimal solution $x$ of the residual LP, 
removes edges $e$ with $x_e=0$, and moves from $F$ to $I$ edges $e$ with $x_e=1$. 
If no such edges exist, it tries to find a set $C$ such that 
all edges with both ends in $C$ and at least $k-2$ edges  
with one end in $C$ are integral. 
If such $C$ is found, the cut constraint of $C$ is dropped from the LP and $C$ is contracted.
Else, it is shown that there exists a pair of contracted nodes $u,v$ with almost $k/2$ integral $uv$-edges. 
The algorithm then adds a ``ghost edge'' between $u$ and $v$, 
relaxing the constraints of  all $uv$-cuts;
during the algorithm the ghost edges are treated as ordinary edges, 
but they do not appear in the returned solution.
Summarizing, cut constraints are relaxed by three operations: 
dropping the cut constraints of a set $C$ with low fractionality, 
contracting such $C$, and adding ghost edges.
This algorithm may add parallel ghost edges, and gives connectivity guarantee $k-10$.
Nutov and Cohen \cite{CN} showed that adding parallel ghost edges can be avoided,
and improved the connectivity guarantee to $k-4$.

Our approach is simpler -- when $k$ is even, we just drop the constraints of all cuts 
that accumulated at least $k-2$ integral edges.
In every iteration we remove from $F$ edges $e$ with $x_e=0$, 
or move from $F$ to $I$ edges $e$ with $x_e=1$, 
where the later operation also may lead to relaxing some cut constraints.
Note that we never round, but just relax some cut constraints.
For $k$ odd, we just run the same algorithm with $k$ replaced by $k'=k-1$. 

However, for this approach to work we need to prove that any extreme point solution $x$
always has an edge $e$ such that $x_e \in\{0,1\}$. 
For this, we will show that there exists a laminar family 
$\LL$ of tight sets such that $x$ is the unique solution to the corresponding equation system.
Unfortunately, a standard uncrossing 
by intersection-union $A \cap B, A \cup B$ or 
by the differences $A \sem B,B \sem A$
does not suffice here. 
We will isolate one case where we will need 
to uncross with two ``non-standard'' sets, say $A \cup B,B \sem A$, 
see Lemma~\ref{l:two} and Fig.~\ref{f:cases}(c). 

To understand the difficulty of the problem it is also worth to ask -- 
what was the best known bicriteria approximation of the form $(1,k-p)$ for unit costs?
Gabow, Goemans, Tardos and Williamson \cite{GGTW},
showed that {\sc Min Size $k$-ECSS} achieves approximation ratios 
$1+2/k$ for $k$ even and $1+3/k$ for $k$ odd w.r.t. the Cut-LP.
This was improved by Gabow and Gallagher \cite{GG} to
$1+\f{21}{11k}$ for (multi-)graphs, and 
$1+\f{1}{k}+\f{q}{k^2}$ for simple graphs for some constant $q$, 
To obtain a $(1,k-p)$ bicriteria approximation, 
take an optimal LP solution $x$ to {\sc $k$-ECSS}, 
and apply the \cite{GG} approximation algorithm on the vector $\f{k-p}{k}x$,
which is a feasible LP solution to {\sc $(k-p)$-ECSS}
(note that in {\kECSS} we can scale down but not up, due to the constraints $x_e \le 1$).
In the multi-graph case the solution size is bounded by
$h(p) =\left(1-\f{p}{k} \right) \cdot \left( 1+\f{21}{11k} \right)=1+\f{21-11p}{11k}-\f{21p}{11k^2}$.
Then we get $h(p) \le 1$ if $\f{21 p}{k} \ge 21-11 p$, which holds for any $k$ if $p=2$. 
In the case of simple graphs, the bound is 
$h(p) =\left(1-\f{p}{k} \right) \cdot \left(1+\f{1}{k}+\f{q}{k^2}\right)$ and then,
by the same method, we get a $(1, k-1)$ bicriteria approximation for $k \ge 2q+1$;
but note again that this is so only for simple graphs.
So to the best of our knowledge, $p=2$ was the best known value of $p$ for unit costs.

We now survey some related work.
{\kECSS} is an old fundamental problem in combina\-torial optimization and network design.
The case $k=1$ is the {\sc MST} problem, while
Pritchard \cite{P} showed that  there is a constant $\epsilon > 0$ such that for any $k \ge 2$,
{\kECSS} admits no $(1+\epsilon)$-approximation algorithm, unless P = NP.
No such hardness of approximation was known for {\kECSM},
and Pritchard \cite{P} conjectured that {\kECSM} admits approximation ratio $1+O(1/k)$.
This conjecture was resolved by Hershkowitz, Klein, and Zenklusen \cite{HKZ} 
who also proved a matching hardness of approximation result.

The {\kECSS} problem and its special cases have been extensively
studied, c.f \cite{Lov,F-aug,FJ1,FJ2,KV,CT,CL,K,GGTW, GG,P, SV,TZ,HKZ} for only a small sample of papers in the field. 
Nevertheless, the current best known approximation ratio for {\kECSS} is still $2$, the same as was four decades ago. 
The $2$-approximation is obtained by bidirecting the edges of $G$, 
computing a min-cost directed spanning subgraph that contains $k$ edge-disjoint dipaths 
to all nodes from some root node (this can be done in polynomial time \cite{E}), 
and returning its underlying graph.
A $2$-approximation can also be achieved with the iterative rounding method \cite{J},
that gives a $2$-approximation also for the more general {\sc Steiner Network} problem, 
where we seek a min-cost subgraph that contains 
$r_{uv}$ edge-disjoint paths for every pair of nodes $u,v \in V$.
The directed version of {\kECSS} was also studied, c.f. \cite{E,GGTW,CT,K}, 
and for this case also no better than $2$-approximation is known, except for special cases. 

Bicriteria approximation algorithms were widely studied for 
degree constrained network design problems such as {\sc Bounded Degree MST} and {\sc Bounded Degree Steiner Network},
c.f. \cite{SL,LZ,LRS} and the references therein.
A particular case of a bicriteria approximation is when only one parameter is relaxed
while the cost/size of the solution is bounded by a budget $B$.
For example, the $(1-1/e)$-approximation for the {\sc Budgeted Max-Coverage} problem 
can be viewed as a bicriteria approximation $(1, 1-1/e)$ for {\sc Set Cover}, 
where only a fraction of $1-1/e$ of the elements is covered by ${\sf opt}$ sets.  
Similarly, in the budgeted version {\sc Budgeted ECSS} of {\kECSS}, instead of $k$ we are given a budget $B$ 
and seek a spanning subgraph of cost at most $B$ that has maximum edge-connectivity $k^*$.

One can view  Theorem~\ref{t:main} as a $k^*-3$ ($k^*-2$ for $k$ even) 
additive approximation for {\sc Budgeted ECSS}, where 
$k^*$ is the maximum edge-connectivity under the budget $B$.
We note that budgeted connectivity problems were studied before for unit costs.
Specifically, in the {\sc Min-Size \kECSS} and {\sc Min-Size \kCSS} problems 
we seek a $k$-edge-connected and $k$-vertex-connected, respectively, spanning subgraph with a minimal number of edges. 
Nutov \cite{N-small} showed that the following simple heuristic (previously analyzed by Cheriyan and Thurimella \cite{CT}) 
is a $(1,k-1)$ bicriteria approximation for {\sc Min-Size \kCSS}:
compute a min-size $(k-2)$-edge-cover $I_{k-2}$ (a spanning subgraph of minimal degree $\ge k-2$) 
and then augment it by an inclusion minimal edge set $F \subs E$ such that $I_{k-2} \cup F$ is $(k-1)$-connected. 
This $(1,k-1)$ bicriteria approximation implies a $k^*-1$ approximation for {\sc Budgeted Min-Size \kCSS},
which is tight, since the problem is NP-hard.

This paper is organized as follows. 
In the next Section~\ref{s:algo} we describe the algorithm 
and show that it can be implemented in polynomial time.
In Section~\ref{s:main} we prove the main theorem about  extreme points of the LP.
Section~\ref{s:main'} describes the $(3/2,k-1)$ bicriteria approximation.
Section~\ref{s:conc} contains some concluding remarks. 

\section{The Algorithm} \label{s:algo}

For an edge set $F$ and disjoint node sets $S,T$ let $\de_F(S,T)$ denote the set of edges in $F$ with 
one end in $S$ and the other end in $T$, and let $d_F(S,T)=|\de_F(S,T)|$.
Let $\de_F(S)=\de_F(S,\bar{S})$ and $d_F(S)=|\de_F(S)|$,
where $\bar{S}=V \sem S$ is the node complement of $S$. 
For $x \in \mathbb{R}^E$ let $x(F)=\sum_{e \in F} x(e)$. 
We will often refer to a proper subset $S$ of $V$ (so $\empt \ne S \subset V$) as a {\bf cut}. 

As in \cite{HKZ,CN}, we will use the iterative relaxation method 
(that previously was used for the {\sc Degree Bounded MST} problem \cite{SL}).
This means that we maintain a set $I$ of integral edges and a set $F$ of fractional edges, 
and while $F \ne \empt$, repeatedly compute an extreme point optimal solution $x$ to the Cut-LP 
and do at least one of the following steps:
\begin{enumerate}
\item
Remove from $F$ edges $e$ with $x_e=0$.
\item
Move from $F$ to $I$ edges $e$ with $x_e=1$.
\item
Relax some cut constraints $x(\de_{F}(S))  \ge k-d_I(S)$ to $x(\de_{F}(S)) \ge k-d_I(S)-q(S)$.  
\end{enumerate}

It would be convenient to consider the residual constraints $x(\de_F(S)) \ge f(S)$ 
for an appropriately defined set function $f$
defined on proper subsets of $V$ by  
\begin{equation} \label{f}
f(S) = \left \{ \begin{array}{ll}
k-d_I(S)-2 \ \ & \mbox{if } S \in \RR \\
k-d_I(S)        & \mbox{otherwise}
\end{array} \right .
\end{equation}
and $f(\empt)=f(V)=0$.
We will assume that $k$ is even and let 
\begin{equation} \label{e:RR}
\RR=\{S:d_I(S) \ge k-2\} \ .
\end{equation}
This function $f(S)$ is a relaxation of the usual residual function $k-d_I(S)$ of {\kECSS}.
When the number $d_I(S)$ of integral edges in a cut $S$ is large enough -- at least $k-2$, 
we subtract $2$ from the demand of $S$, 
which is equivalent to dropping the constraint of $S$. 
We thus consider the following LP-relaxation with $f$ as in (\ref{f}) 
where $\RR$ given in (\ref{e:RR}).
\begin{equation} \label{LP}
\displaystyle
\begin{array} {lllllll} 
& \min             & c^{T} \cdot x              &                                                  \\ 
& \ \mbox{s.t.} & x(\de_F(S)) \ge f(S) & \forall \empt \ne S \subset V  \\
&                      & 0 \le x_e \le 1           & \forall e \in F    
\end{array}
\end{equation}

Algorithm~\ref{alg:main} below is for $k$ even.
For $k$ odd, we run the same algorithm with $k$ replaced by $k'=k-1$. 
The connectivity guarantee then is $k'-2=k-3$.  
It is easy to see that for any $k' \le k-1$, if $x$ is a feasible solution to LP (\ref{LPk})
then $\f{k'}{k}x$ is a feasible solution to for this LP with $k'$. 
This implies that the optimal LP value for $k'$ is at most $k'/k$ the optimal value for $k$, 
and for $k'=k-1$ this gives cost approximation $\f{k-1}{k}=1-\f{1}{k}$. 

\medskip 

\begin{algorithm}[H]
\caption{A bicriteria $(1,k-2)$-approximation for $k$ even} \label{alg:main}
{\bf initialization}: $I \gets \empt$, $F \gets E$ \\
\While{$F \ne \empt$}
{
compute an extreme point solution $x$ to LP  (\ref{LP}) with $f$ in (\ref{f}) and $\RR$ in (\ref{e:RR}) \\
remove from $F$ every edge $e$ with $x_e=0$ \\
move from $F$ to $I$ every edge $e$ with $x_e=1$ 
}
\Return{$I$}
\end{algorithm}

\medskip 

In Section~\ref{s:main} we will prove the following lemma about extreme point solutions of LP (\ref{LP}).

\begin{lemma} \label{l:main}
Let $x$ be an extreme point of the polytope 
$P=\{x \in [0,1]^F:x(\de_F(S)) \ge f(S)\}$ where $k$ is even,
$f$ is defined in (\ref{f}), 
$\RR$ is defined in (\ref{e:RR}), and 
$F \ne \empt$. 
Then there is $e \in F$ such that $x_e \in \{0,1\}$. 
\end{lemma}

Assuming Lemma~\ref{l:main}, we show that the algorithm can be implemented in polynomial time.

\begin{lemma} \label{l:m}
The algorithm terminates after at most $2n$ iterations, where $n=|V|$.
\end{lemma}
\begin{proof}
Let $x$ be an extreme point computed at the first iteration. 
It is known (c.f. \cite{GGTW}) that $|\{0<x_e <1:e \in F\}| \le 2n-1$, namely, at most $2n-1$ variables $x_e$ are fractional.
This follows from two facts. 
The first is that after the integral entries are substituted, the fractional entries are determined 
by a full rank system of constraints $\{x(\de_F(S))=k:S \in \LL\}$ where $\LL$ is laminar, c.f. \cite{J,GGTW}. 
The second is that a laminar family on a set of $n$ elements has size at most $2n-1$.
Thus already after the first iteration, $|F| \le 2n-1$. 
At every iteration an edge is removed from $F$, thus the number of iteration is at most $2n$.
\end{proof}

To show a polynomial time implementation it is sufficient to prove the following.

\begin{lemma} \label{l:impl}
At any step of the algorithm, an extreme point solution $x$ to LP  (\ref{LP}) 
with $f$ in (\ref{f}) and $\RR$ in (\ref{e:RR})
can be computed in polynomial time. 
\end{lemma}
\begin{proof}
An extreme point optimal solution can be found using the Ellipsoid Algorithm.
For that, we need the show an existence of a polynomial time separation oracle: 
given $x$, determine whether $x$ is a feasible LP solution or find a violated inequality.
Checking the inequalities $0 \le x_e \leq 1$ is trivial.
An inequality of a cut $S$ is violated if 
\[
d_I(S) \le k-3 \ \  \mbox{ and } \ \ d_I(S)+x(\de_F(S)) <k \ .
\]
For $k=2$ there is nothing to check. 
For $k=3$ the condition reduces to $d_I(S)=0$ and $x(\de_F(S)) <3$, which can be checked 
by computing a minimum cut in the graph obtained 
by contracting each connected component of $(V,I)$ and assigning 
capacities $x_e$ to the edges in $F$.
For $k \ge 4$ we first check if there exists $S$ such that $d_I(S)+x(\de_F(S)) <k-2$ by assigning 
capacity $1$ to edges in $I$, capacities $x_e$ to edges in $F$, and computing a minimum cut.
If such $S$ exists, then its inequality is violated. 
Assume therefore that $k \ge 4$ and that $d_I(S)+x(\de_F(S)) \ge k-2$ for all $S$,
and let ${\cal S}$ be the family of cuts whose capacity is in the range $[k-2,k)$.
Then the problem boils down to checking whether there exists $S \in {\cal S}$
such that $d_I(S) \le k-3$, where the capacity of a minimum cut is $k-2$.
It is known that for any constant $\alpha$ the number of $\alpha$-near minimum cuts 
is $O\left(n^{\lfloor 2\alpha \rfloor}\right)$,
and that they all can be listed in polynomial time \cite{Kar} (see also \cite{BCW}).
Note that $\f{k}{k-2} \le 2$ for $k \ge 4$, thus 
we can list all cuts in ${\cal S}$ in polynomial time, 
and check for each $S \in {\cal S}$ whether $d_I(S) \le k-3$ or not.
\end{proof}

It remain to prove Lemma~\ref{l:main}, which we will do in the next section.

\section{Integrality of extreme points (Lemma~\ref{l:main})} \label{s:main}

Here we will prove Lemma~\ref{l:main}. 
Let $x$ be an extreme point as in Lemma~\ref{l:main}.
We say that a set $S$ is {\bf tight} if 
the LP-inequality of $S$ holds with equality, namely, if  $x(\de_F(S))=f(S)$. 
Let $\TT=\{S: x(\de_F(S))=f(S)>0\}$ be the the set of all tight sets with strictly positive demand. 
Assume by contradiction that $0<x_e<1$ for all $e \in F$. 
Then there exists a family $\LL \subs \TT$ such that 
$x$ is the unique solution to the equation system $\{x(\de_F(S))=f(S):S \in \LL\}$; namely, 
$|\LL|=|F|$ and the characteristic vectors of the sets $\{\de_F(S):S \in \LL\}$ are linearly independent.  
We call such a family $\LL$ {\bf $x$-defining}. We will prove that if $x$ as above exists,  
then the following holds. 

\begin{lemma} \label{l:lami}
There exists an $x$-defining family that is laminar.
\end{lemma}

We will now show that Lemma~\ref{l:lami} implies a contradiction 
to the assumption that $0<x_e<1$ for all $e \in E$.
The following was essentially proved by Jain \cite[Theorem 4]{J}, 
see also \cite[Exercise 23.4]{V} and \cite{HKZ,CN}.

\begin{lemma}[\cite{J}] \label{l:J3}
Let $x$ be an extreme point solution for LP (\ref{LP}) with 
any integral set function $f$ 
such that $0<x_e<1$ for all $e \in E$, and suppose that 
there exists an $x$-defining family $\LL$ that is laminar.
Then there is an inclusion-minimal set $C$ in $\LL$ such that $d_F(C) \le 3$. 
\end{lemma}

By Lemma~\ref{l:lami} there exists an $x$-defining laminar family $\LL$, 
while by Lemma~\ref{l:J3} 
there is $C \in \LL$ such that $d_F(C) \le 3$;
this implies $x(\de_F(C)) \le 2$ and $d_I(C) \ge k-2$ 
(by the integrality of $f$),
hence $C \in \RR$, contradicting that $f(C)>0$.
Thus it remains only to prove Lemma~\ref{l:lami}, which we will do in the rest of this section.

\medskip

Two sets $A,B$ {\bf overlap} if $A \cap B,A \sem B,B \sem A$ are nonempty;
if also $A \cup B \ne V$ then $A,B$ {\bf cross}.
For a cut $S$ let $\chi_S \in \{0,1\}^F$ be the incidence vector of $\de_F(S)$.
We say that overlapping $A,B \in \TT$ are {\bf $x$-uncrossable} 
if each one of $\chi_A,\chi_B$ can be expressed as a linear combination 
of the characteristic vector of the other and 
the characteristic vectors of one or two sets in $\TT \cap \{A \cap B,A\cup B,A \sem B, B \sem A\}$. 
The following was essentially proved by Jain \cite[Lemma 4.2]{J}, see also \cite[Theorem 23.12]{V}.

\begin{lemma}[\cite{J}]
If any $A,B \in \TT$ are $x$-uncrossable, then there exists an $x$-defining family that is laminar.
\end{lemma}
\begin{proof}
Let $\LL$ be a laminar subfamily of $\TT$ whose incidence vectors are linearly independent.
Let $\span(\TT)$ denote the linear space spanned by the incidence vectors of the sets in $\TT$ 
and similarly $\span(\LL)$ is defined. If $\span(\LL)=\span(\TT)$ then we are done. 
Else, let $S \in \TT$ be a set with $\chi_S \notin \span(\LL)$ that overlaps the minimal number of sets in $\LL$.
We claim that $\LL \cup \{S\}$ is laminar.
Suppose to the contrary that $S$ overlaps some $T \in \LL$. 
By the assumption, $\chi_S$ can be expressed as a linear combination 
of $\chi_T$ and the characteristic vectors of one or two sets $X,Y \in \{S \cap T, S \cup T, S \sem T, T \sem S\} \cap \TT$
(possibly $X=Y$).
Note that $X,Y$ cannot be both in $\span(\LL)$, since otherwise $S \in \span(\LL)$. 
But each of $S \cap T, S \cup T, S \sem T, T \sem S$ overlaps strictly less sets in $\LL$ than $S$, c.f. \cite[Lemma 23.15]{V}.
This contradicts our choice of $S$. 
\end{proof}

Thus to prove Lemma~\ref{l:lami}, it is sufficient to prove the following.

\begin{lemma} \label{l:uncross}
Any $A,B \in \TT$ are $x$-uncrossable
\end{lemma}

Assume that $A,B \in \TT$ cross.
Otherwise, $A \sem B=\bar{B}$, $B \sem A=\bar{A}$, 
and thus $\chi_{A \sem B}=\chi_{\bar{B}}=\chi_B$ and $\chi_{B \sem A}=\chi_{\bar{A}}=\chi_A$,
implying that $A \sem B,B \sem A$ are both tight and  $\chi_A+\chi_B=\chi_{A \sem B}+\chi_{B \sem A}$. 
The next Lemma uses a standard uncrossing by $A \cap B,A \cup B$ or by $A \sem B,B \sem A$.

\begin{lemma} \label{l:diag}
If $A \cap B,A \cup B \notin \RR$ or if $A \sem B, B \sem A \notin \RR$ then $A,B$ are $x$-uncrossable.
\end{lemma}
\begin{proof}
We prove that if $A \cap B,A \cup B \notin \RR$ then $A,B$ are $x$-uncrossable;
the proof of the case $A \sem B,B \sem A \notin \RR$ is similar and in fact 
can be deduced from the first part by the symmetry of $f$. 
Denoting $\De=x(\de_F(A \sem B,B \sem A))$ we have:
\begin{eqnarray*}
f(A)+f(B) & = &    x(\de_F(A))+x(\de_F(B)) 
                  =       x(\de_F(A \cap B))+x(\de_F(A \cup B))+2\De \\
			   			& \ge & f(A \cap B)+f(A \cup B)+2 \De 
					  	   \ge     f(A)+f(B)+2\De \ .
\end{eqnarray*}
The first equality is since $A,B$ are tight, and 
the second can be verified by counting the contribution of each edge to both sides. 
The first inequality is since $x$ is a feasible LP-solution. 
The second inequality is since none of $A,B,A \cap B,A \cup B$ belongs to $\RR$,
hence $f$ coincides on these sets with the symmetric supermodular function $g(S)=k-d_I(S)$,
that satisfies the supermodular inequality $g(A \cap B)+g(A \cup B) \ge g(A)+g(B)$. 
Consequently, equality holds everywhere, hence $A \cap B,A \cup B$ are both tight.
Moreover, $\De=0$, and this implies $\chi_A+\chi_B=\chi_{A \cap B}+\chi_{A \cup B}$,
concluding the proof.
\end{proof}

From this point it would be convenient to use the following notation, see Fig.~\ref{f:cases}(a). 
The {\bf corner sets} of $A,B$. are 
\[
C_1=A \cap B \ \ \ \ \ C_2=A \sem B \ \ \ \ \  C_3=V \sem (A \cup B) \ \ \ \ \ C_4=B \sem A \ .
\]
We say that $C_i,C_j$ are {\bf adjacent} if $j=i+1$ where the indexes are modulo $4$;
namely, the only non-adjacent corner set pairs are $C_1,C_3$ and $C_2,C_4$.
Note that  $\chi_{C_3}=\chi_{A \cup B}$.

\begin{figure} \centering \includegraphics[scale=0.5]{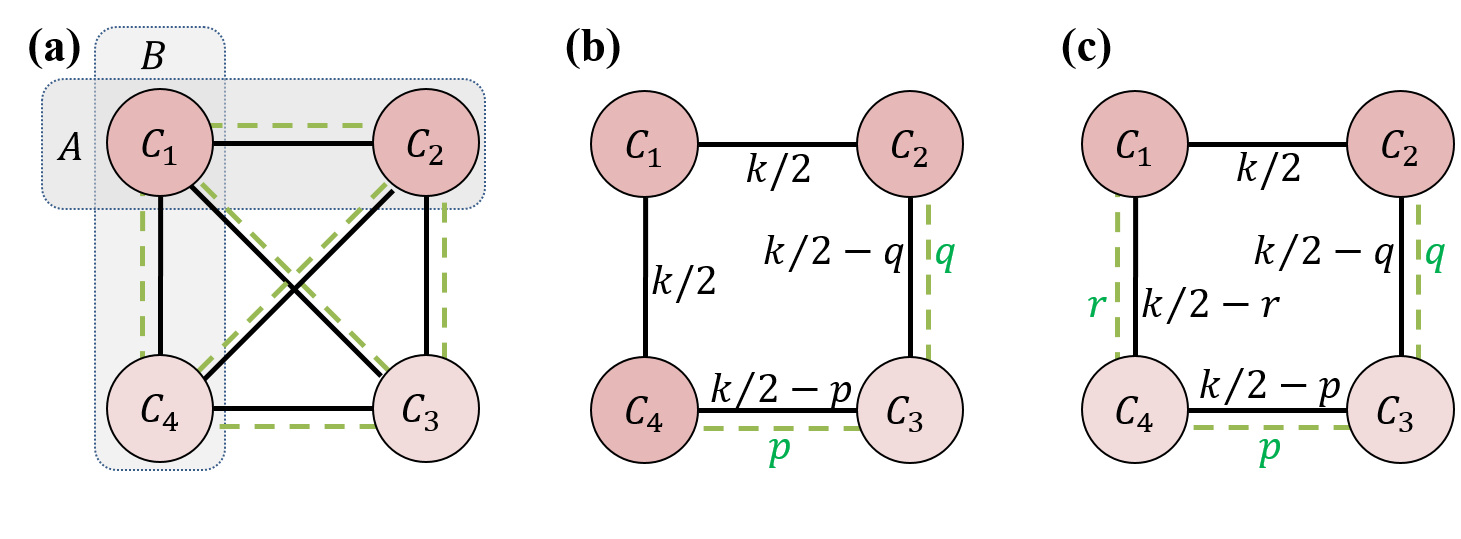}
\caption{
Illustration to the proofs of Lemmas \ref{l:mu}, \ref{l:three}, and \ref{l:two}.
Sets in $\RR$ are shown by darker circles. 
The integral edges (the edges in $I$) are represented  by black lines and 
the fractional edges (the edges in $F$) are shown by green dashed lines;
note that one line may represent many edges.  
The numbers of appropriate color represent the total capacity of fractional or integral edges between sets.  
           }
\label{f:cases} \end{figure}

The next lemma uses for the first time our assumption that $k$ is even.

\begin{lemma} \label{l:mu}
If $k$ is even then $d_I(C_i,C_j) \ge k/2$ for any adjacent $C_i,C_j \in \RR$.
\end{lemma}
\begin{proof}
W.l.o.g. assume that $C_1,C_2 \in \RR$.
Then $d_I(C_1) \ge k-2$ and $d_I(C_2) \ge k-2$.
Since $A \notin \RR$, $d_I(A) \le k-3$. Thus we get (see Fig.~\ref{f:cases}(a)):
\[
k-3 \ge d_I(A)=d_I(C_1)+d_I(C_2)-2d_I(C_1,C_2) \ge 2(k-2)-2d_I(C_1,C_2) \ .
\]
Consequently, $2d_I(C_1,C_2) \ge k-1$, implying $d_I(C_1,C_2) \ge k/2$, since $k$ is even.
\end{proof}

\begin{remark*}
Note that before Lemma~\ref{l:mu} we didn't use the assumption that $k$ is even.
For $k$ even Lemma~\ref{l:mu} gives the bound $d_I(C_i,C_j) \ge \f{k}{2}$,
but for $k$ odd only the bound $d_I(C_i,C_j) \ge \f{k-1}{2}$,
which is less than $k/2$. 
The next two lemmas, apart of using standard uncrossing methods, 
rely only on the fact that there are at least $k/2$ integral between any adjacent sets $C_i,C_j \in \RR$.
If this was true for $k$ odd, then we could also prove Lemma~\ref{l:uncross} for $k$ odd. 
\end{remark*}

\medskip

For disjoint sets $S,T$ let 
$\psi(S,T)=d_I(S,T)+x(\de_F(S,T))$, 
and let $\psi(S)=\psi(S,\bar{S})$ be the ``coverage'' of $S$ by both $I$ and $F$.
Note that $\psi(S)<k$ implies that $S \in \RR$ and that $d_F(S) \ge 2$ for all $S \in \TT$ 
(since $x(\de_F(S)=f(S) \ge 1$ for all $S \in \TT$). 

\begin{lemma} \label{l:three}
At most $2$ corner sets belong to $\RR$.
\end{lemma}
\begin{proof}
If all corner sets are in $\RR$, 
then since $\psi(A)=\psi(B)=k$ and by Lemma~\ref{l:mu}, 
we get that 
$\psi(C_i,C_j)=d_I(C_i,C_j)=k/2$ if $C_i,C_j$ are adjacent 
and $\psi(C_i,C_j)=0$ otherwise. 
But then $d_F(A)=d_F(B)=0$, which is a contradiction.

Suppose now that exactly $3$ corner sets belong to $\RR$,
say w.l.o.g $C_1,C_2,C_4 \in \RR$ and $C_3 \notin \RR$.
Note that:
\begin{itemize}
\item
$\psi(C_3) \ge k$, since $C_3 \notin \RR$.
\item 
$d_I(C_1,C_2) \ge k/2$ and $d_I(C_1,C_4) \ge k/2$, by Lemma~\ref{l:mu}; hence $\psi(C_1) \ge d_I(C_1) \ge k$.
\item
$\psi(C_1)+\psi(C_3)=\psi(A)+\psi(B)-2\psi(C_2,C_4) =2k-2\psi(C_2,C_4)$.
\end{itemize}
From this we get that $\psi(C_1)=\psi(C_3)=k$ and $\psi(C_2,C_4)=0$.
Moreover, $d_I(C_1,C_2)=d_I(C_1,C_4) = k/2$ must hold, 
and thus $\psi(C_1,C_3)=0$. 
Consequently, the situation is as depicted in Fig.~\ref{f:cases}(b), where $p,q$ are integers. 
We must have $q \le 2$ since $C_2 \in \RR$ and we must hace $q \ge 3$ since $A \notin \RR$,
which is a contradiction.
\end{proof}

\begin{lemma} \label{l:two}
If exactly $2$ corner sets belong to $\RR$ then $A,B$ are $x$-uncrossable. 
\end{lemma}
\begin{proof}
The case when two non-adjacent corner sets are not in $\RR$ was proved in Lemma~\ref{l:diag}.
So assume w.l.o.g. that $C_1,C_2 \in \RR$ and $C_3,C_4 \notin \RR$.
We will show that then the following holds, see Fig.~\ref{f:cases}(c):
\begin{enumerate}[(i)]
\item
$d_F(C_1,C_2)=0$ and $\psi(C_1,C_3)=\psi(C_2,C_4)=0$. 
\item
$C_3,C_4 \in \TT$.
\end{enumerate}
Note that (i) implies $\chi_A+2\chi_B=\chi_{C_3}+\chi_{C_4}=\chi_{A \cup B}+\chi_{B \sem A}$ (see Fig.~\ref{f:cases}(c)),
and thus by (ii) $A,B$ are $x$-uncrossable. 

Since $\psi(B)=k$ and $\psi(C_1,C_2) \ge d_I(C_1,C_2) \ge k/2$ (by Lemma~\ref{l:mu})
we have:
\[
\psi(C_3,C_4) \le \psi(B)-\psi(C_1,C_2)  \le k-d_I(C_1,C_2) \le k/2 \ .
\]
Since $\psi(A) =k$ and $\psi(C_3),\psi(C_4) \ge k$ (since $C_3,C_4 \notin \RR$) we have
\[
2\psi(C_3,C_4)=\psi(C_3)+\psi(C_4)-\psi(A) \ge k \ .
\]
Consequently, $\psi(C_3,C_4)=k/2$, and this implies that in the above inequalities equality holds everywhere.
In particular we get:
\begin{itemize}
\item
$\psi(C_3,C_4) = \psi(B)-\psi(C_1,C_2)$ and hence  $\psi(C_1,C_3)=\psi(C_2,C_4)=0$.  
\item
$\psi(C_1,C_2)=d_I(C_1,C_2)=k/2$ hence $d_F(C_1,C_2)=0$.
\end{itemize}

Now we claim that $\psi(C_1,C_4)=\psi(C_2,C_3)=k/2$, see fig.~\ref{f:cases}(c).
This easily follows from the fact that 
$\psi(C_1,C_4)+\psi(C_2,C_3) = \psi(A)=k$ and the following equalities:
\[
\psi(C_1,C_4)+k/2 = \psi(C_1,C_4)+\psi(C_3,C_4)= \psi(C_4) \ge k    
\]
\[
\psi(C_2,C_3)+k/2 = \psi(C_2,C_3)+\psi(C_3,C_4)= \psi(C_3) \ge k    
\]
Consequently, we get that $C_3,C_4$ are tight.
Moreover, since $\de_F(A)$ and $\de_F(B)$ are both nonempty, so are $\de_F(C_3)$ and $\de_F(C_4)$.
This implies that $C_3,C_4 \in \TT$, concluding the proof. 
\end{proof}

The final uncrossing case is illustrated in Fig.~\ref{f:cases}(c),
where $p,q,r$ are positive integers such that:
$q+r \ge 3$ (since $A \notin \RR$), 
$p+q,p+r \ge 3$ (since $C_3,C_4 \notin \RR$),
$p \ge 3$ (since $B \notin \RR$), 
and $q,r \le 2$ (since $C_1,C_2 \in \RR$).
For example, we may have $p=3$ and $q=r=2$. 

This concludes the proof of Lemma~\ref{l:main} and thus also the first part of Theorem~\ref{t:main}. 

\section{A (3/2,{\em k} -- 1) bicriteria approximation} \label{s:main'}

To achieve $(3/2,k-1)$ bicriteria approximation we use the function $f$ defined on proper subset of $V$ by
\begin{equation} \label{f'}
f(S) = \left \{ \begin{array}{ll}
k-d_I(S)-1 \ \ & \mbox{if } S \in \RR \\
k-d_I(S)        & \mbox{otherwise}
\end{array} \right .
\end{equation}
and $f(\empt)=f(V)=0$.
Here the set $\RR$ is defined by: 
\begin{equation} \label{e:RR'}
\RR=\{S:d_I(S) \ge k-1\} \ .
\end{equation}

The main difference between Algorithm~\ref{alg:main} and the following Algorithm~\ref{alg:main'} are:
\begin{enumerate}
\item
We use a different function $f$ -- the one defined in (\ref{f'}). 
\item
We round to $1$ edges $e$ with $x_e \ge 2/3$. 
\item
The algorithm is for both $k$ odd and even.
\end{enumerate}

\medskip 

\begin{algorithm}[H]
\caption{A bicriteria $(3/2,k-1)$-approximation} \label{alg:main'}
{\bf initialization}: $I \gets \empt$, $F \gets E$ \\
\While{$F \ne \empt$}
{
compute an extreme point solution $x$ to LP  (\ref{LP}) with $f$ in (\ref{f'}) and $\RR$ in (\ref{e:RR'}) \\
remove from $F$ every edge $e$ with $x_e=0$ \\
move from $F$ to $I$ every edge $e$ with $x_e \ge 3/2$ 
}
\Return{$I$}
\end{algorithm}

\medskip 

Note that here we combine iterative relaxation with iterative rounding.
Since in each iteration the cut constraints are only relaxed, 
and only edges $e$ with $x_e \ge 2/3$ are rounded,
the solution cost is at most $3/2$ times the initial LP-value.
The connectivity guarantee $k-1$ is also straightforward.
We will prove the following counterpart of Lemma~\ref{l:lami}.

\begin{lemma} \label{l:lami'}
Let $x$ be an extreme point of the polytope 
$P=\{x \in [0,1]^F:x(\de_F(S)) \ge f(S)\}$ where 
$f$ is defined in (\ref{f'}), 
$\RR$ is defined in (\ref{e:RR'}), and 
$F \ne \empt$. 
Suppose that $0<x_e<1$ for all $e \in F$.
Then there exists an $x$-defining family $\LL$ that is laminar.
\end{lemma}
\begin{proof}
As before, it is sufficient to prove that any $A,B \in \TT$ that cross are $x$-uncrossable,
where  $\TT=\{S: x(\de_F(S))=f(S)>0\}$. Fix $A,B \in \TT$ that cross. 
By a totally identical proof to that of Lemma~\ref{l:diag} 
(standard uncrossing by $A \cap B,A \cup B$ or by $A \sem B,B \sem A$)
we have:

\begin{claim*}
If $A \cap B,A \cup B \notin \RR$ or if $A \sem B, B \sem A \notin \RR$ then $A,B$ are $x$-uncrossable. \hfill $\Box$
\end{claim*}

The following claim (a counterpart of Lemma~\ref{l:mu}) holds both for $k$ odd an even.

\begin{claim*} 
$d_I(C_i,C_j) \ge \lceil k/2 \rceil$ for any adjacent $C_i,C_j \in \RR$.
\end{claim*}
{\em Proof.}
W.l.o.g. assume that $C_1,C_2 \in \RR$.
Then $d_I(C_1) \ge k-1$ and $d_I(C_2) \ge k-1$.
Since $A \notin \RR$, $d_I(A) \le k-2$. Thus we get 
\[
k-2 \ge d_I(A)=d_I(C_1)+d_I(C_2)-2d_I(C_1,C_2) \ge 2(k-1)-2d_I(C_1,C_2) \ .
\]
Consequently, $2d_I(C_1,C_2) \ge k$, implying $d_I(C_1,C_2) \ge \lceil k/2 \rceil$.
\hfill $\Box$

\medskip

As was mentioned in the remark after Lemma~\ref{l:mu}, 
the proofs of Lemma \ref{l:three} and \ref{l:two} 
rely only on the fact that there are at least $k/2$ edges between any adjacent $C_i,C_j \in \RR$.
Since this condition holds in our case for both $k$ odd and even, 
Lemmas \ref{l:three} and \ref{l:two} also hold, concluding the proof. 
\end{proof}

We will now show that the algorithm terminates, namely, that there always exists an edge $e$ with 
$x_e=0$ or $x_e \ge 2/3$;
other polynomial time implementation details are identical to those in Lemmas~\ref{l:m}~and~\ref{l:impl}.
To show that such $e$ always exists, it is sufficient to prove the following counterpart of Lemma~\ref{l:main}. 

\begin{lemma} \label{l:main'}
Let $x$ be an extreme point as in Lemma~\ref{l:lami'}.
Then there is $e \in F$ such that $x_e \ge 2/3$. 
\end{lemma}
\begin{proof}
Let $\LL$ be a laminar $x$-defining family as in Lemma~\ref{l:lami'}. 
By Lemma~\ref{l:J3},
there is $C \in \LL$ such that $d_F(C) \le 3$.
Since $C$ is tight, $f$ is integral, and $0<x_e<1$ for all $e \in F$, we must have 
$d_I(C)=k-1$ or $d_I(C)=k-2$. 
The former case is not possible since $d_I(C)=k-1$ implies $C \in \RR$, 
while in the latter case $x_e \ge2/3$ for some $e \in \de_F(C)$, as claimed.
 \end{proof}

This concludes the proof of the second part of Theorem~\ref{t:main}.

\section{Concluding remarks} \label{s:conc}

We showed that {\kECSS} admits bicriteria approximation ratios 
$(1,k-2)$ for $k$ is even and $\left(1-\f{1}{k},k-3\right)$ for $k$ odd.
We also obtained a $(3/2,k-1)$ bicriteria approximation, which is the first algorithm
that achieves $k-1$ connectivity by cost approximation below $2$.
However currently no hardness result indicates that 
approximation $(1,k-1)$ cannot be achieved. 
On the other hand, such a result is known 
only for simple graphs with unit costs -- it can be deduced from the result of Gabow and Gallagher \cite{GG},
and their algorithm is highly involved. 

As far as we can see, our results extend to the more general ``set-connectivity'' problem
$
\min\{c^T \cdot x:x(\de_E(S)) \ge g(S), x \in \{0,1\}^E\}
$
with an arbitrary symmetric integral set function $g$ that can be ''uncrossed'' 
both by intersection-union $A \cap B, A \cup B$ and by the differences $A \sem B,B \sem A$.
One such class are set functions obtained by zeroing a symmetric 
supermodular function $h$ on non-proper sets, 
namely $g(S)=h(S)$ if $S \notin \{\empt,V\}$ and $g(\empt)=g(V)=0$ 
(for {\kECSS}, $h(S)=k$ for all $S$). 
This algorithm will achieve a bicriteria approximation $(1,g(S)-2)$ if $g$ is even-valued.
Unfortunately, we failed to find additional problems, 
besides {\kECSS}, that comply with this setting. 


\end{document}